\documentclass[conference]{IEEEtran}
\usepackage{cite}
\usepackage{amsfonts,amssymb,amsthm,amsbsy,amsmath,paralist,bm,ifthen,color}
\usepackage{algorithm,algorithmic}
\usepackage{graphicx}
\usepackage{subfig}
\usepackage[font=footnotesize]{caption}
\usepackage{xcolor}
\usepackage{relsize}
\usepackage[numbers,sort&compress,square]{natbib}
\usepackage{tabularx}
\usepackage{booktabs}
\usepackage{multirow}
\usepackage{setspace}
\newtheorem{lemma}{Lemma}

\IEEEoverridecommandlockouts
\begin{document}
\captionsetup[figure]{labelfont={bf},name={Fig.},labelsep=period}
\title{Robust Resource Allocation for Multi-Antenna URLLC-OFDMA Systems in a Smart Factory}
\author{Jing Cheng$^\dag$, Chao Shen$^\dag$, Shuqiang Xia$^\ddag$\\
{\small $^\dag$State Key Laboratory of Rail Traffic Control and Safety, Beijing Jiaotong University, Beijing, China}\\
{\small $^\ddag$State Key Laboratory of Mobile Network and Mobile Multimedia Technology, ZTE Corporation, Guangdong, China}\\
{\small Email: \{chengjing, chaoshen\}@bjtu.edu.cn, xia.shuqiang@zte.com.cn}}
%\thanks{The work was supported in part by the NSFC, China under Grant 61871027 and Grant U1834210, in part by the State Key Laboratory of Rail Traffic Control and Safety under Grant RCS2019ZZ002.}}
\maketitle
\begin{abstract}
In this paper, we investigate the worst-case robust beamforming design and resource block (RB) assignment problem for total transmit power minimization of the central controller while guaranteeing each robot's transmission with target number of data bits and within required ultra-low latency and extremely high reliability. By using the property of the independence of each robot's beamformer design, we can obtain the equivalent power control design form of the original beamforming design. The binary RB mapping indicators are transformed into continuous ones with additional $\ell_0$-norm constraints to promote sparsity on each RB. A novel non-convex penalty (NCP) approach is applied to solve such $\ell_0$-norm constraints. Numerical results demonstrate the superiority of the NCP approach to the well-known reweighted $\ell_1$ method in terms of the optimized power consumption, convergence rate and robustness to channel realizations. Also, the impacts of latency, reliability, number of transmit antennas and channel uncertainty on the system performance are revealed.
\end{abstract}
\section{Introduction}
With the evolution into beyond $5$G or $6$G wireless networks, there is a paradigm shift in the primary service objects of networks from mobile phone users to massive low-power and low-cost machines. Such machine-to-machine (M2M) communication, popular in the Internet of things (IoT) network, can be divided into two types: massive IoT applications and critical IoT applications. As compared to massive connections and maximum throughput oriented massive IoT applications, critical IoT applications aim at ultra-reliable and low-latency communication (URLLC) \cite{SchulzICM2017}. More specifically, URLLC is required to transmit short packets (e.g. $32$ bytes) successfully within ultra-low latency (e.g. user-plane latency $1$ ms) and with no less than $99.999\%$ reliability (i.e. $10^{-5}$ packet error probability) \cite{KimPI2019}.

In URLLC scenarios, the coding blocklength is short and the decoding error probability becomes no longer negligible even if the transmission rate is below the Shannon limit. In this case, conventional resource allocation based on Shannon capacity achieved with infinite blocklength codes is not optimal, which necessitates the research on the resource allocation and transmission scheme design under the finite blocklength regime \cite{XuITWC2016}. Meanwhile, the coupling of high-reliability and low-latency renders such design very challenging. In addition, most of the existing works consider such problems generally assuming that the channel state information (CSI) can be perfectly known at the transmitter. For example, \cite{XuITWC2020} investigated the energy-efficient packet transmission problem for a two-user non-orthogonal multiple access (NOMA) downlink with heterogeneous latency constraints under the ideal CSI assumption. In \cite{RenITC2020}, the resource allocation problem for a secure mission-critical IoT communication system with URLLC was studied under strict assumption that all links' CSI is available at the transmitter. Actually, such solutions under ideal CSI assumption can serve as a performance benchmark, but it is not practical especially for URLLC scenarios.

Flexible new radio (NR) frame structure is proposed to support the ultra-low latency requirement  \cite{SachsIN2018}. There comes a critical challenge on resource and latency efficient scheduling in URLLC-orthogonal frequency division multiple access (OFDMA) systems. For example, a global optimal resource allocation scheme for URLLC service was proposed in \cite{SunITWC2019} by jointly optimizing uplink and downlink. However, the global optimal solution was based on the assumption that the channel gain is the same even for different sub-channels. Resource block (RB) assignment and power allocation problem of single-cell multiple URLLC users with perfect and imperfect CSI were examined in \cite{Ghanem2019} and \cite{Cheng2020}, respectively. However, the above works did not take the multiple antennas into account, which plays a key role in improving the link quality and reliability.

In this paper, we consider a smart factory scenario where a central controller has to send critical control commands to its serving robots with strict latency and reliability requirements. Multiple-antenna technique is used in the central controller to enhance the communication reliability. We investigate the total transmit power minimization problem by jointly optimizing RB assignment and power control design in the finite blocklength regime under the imperfect CSI assumption. The main contributions of this paper are summarized as follows.
\begin{itemize}
  \item A robust transmission scheme is proposed. By capturing the property of the independence of each robot's beamformer design, we equivalently transform the original beamforming design into a power control problem. In view of the binary and sparse constraints on each RB, we relax them into continuous variables and add $\ell_0$-norm constraints to guarantee the sparsity.
  \item A novel non-convex penalty (NCP) approach is applied to tackle the $\ell_0$-norm constraints. Further, a low-complexity penalized successive convex approximation (SCA) based iterative algorithm is proposed to efficiently solve the formulated joint RB assignment and power control problem.
  \item Simulation results show the performance superiority of NCP-based resource allocation to the reweighted $\ell_1$ method and analyze the impacts of key factors like latency, reliability, number of transmit antennas and channel uncertainty on the system performance.
\end{itemize}
\section{System and Channel Uncertainty Models}
\subsection{System Model}
Consider a smart factory scenario, where a central controller equipped with $N_t$ antennas has to send critical control commands to $K$ single-antenna robots indexed by $k\in\{1,\ldots,K\}$, as shown in Fig. \ref{systemModel}. The command packet with $B_k$ data bits for the $k$-th robot has to be successfully transmitted within $D_k$ OFDM symbols and with packet error probability $\varepsilon_k,\forall k$. We assume the requirements $\{B_k,D_k,\varepsilon_k\}_{k=1}^K$ are known at the central controller. The fundamental scheduling resource unit is a RB, which is composed of $12$ subcarriers in the frequency domain and $1$ OFDM symbol in the time domain \cite{Ghanem2019}. Denote the number of RBs and OFDM symbols for scheduling by $M$ and $N$, respectively. Since one packet can be transmitted over multiple RBs, we define a binary variable $\phi_{mnk}\in\{0,1\}$ to indicate the RB mapping. If the RB $m$ in OFDM symbol $n$ is allocated to robot $k$, we have $\phi_{mnk}=1$, otherwise $\phi_{mnk}=0$, where $m\in \{1,\ldots,M\},~n\in\{1,\ldots,N\}$. We assume that each RB can be allocated to at most one robot, which can be characterized by
\begin{equation}
  \sum_{k=1}^K \phi_{mnk} \leq 1,\forall m,n.
\end{equation}
\begin{figure}[htbp]\centering
	\includegraphics[width=.66\linewidth]{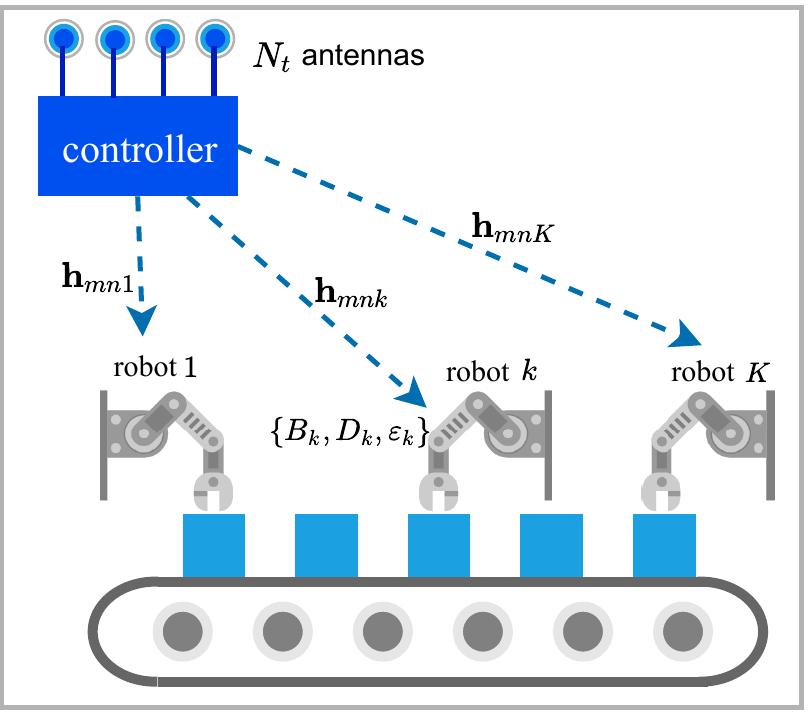}
	\caption{Multi-robot downlink URLLC transmission in a smart factory scenario.}\vspace{-5mm}
	\label{systemModel}
\end{figure}
\subsection{Channel Uncertainty Modeling}
Under the quasi-static block fading channel assumption, the received signal of user $k$ in the RB $m$ and OFDM symbol $n$ can be expressed as
\begin{equation}
  y_{mnk}=\phi_{mnk}\mathbf{h}_{mnk}^H \mathbf{w}_{mnk}d_{mnk}+z_{mnk},
\end{equation}
where $\mathbf{h}_{mnk}\in \mathbb{C}^{N_t}$ is the channel realization, $\mathbf{w}_{mnk}\in \mathbb{C}^{N_t}$ is the transmit beamforming vector, $d_{mnk}$ denotes the independent complex data symbol with power normalized to be unit, and $z_{mnk}$ is the zero mean circular symmetric complex Gaussian noise with variance $\sigma^2$.

However, owing to practical limitations such as channel mismatch and quantization error, perfect CSI is not available at the central controller. This is true especially for the mission-critical scenarios where the  transmission time interval (TTI) is very short and the time for channel training is highly restricted. In such situation, the channel realization $\mathbf{h}_{mnk}$ can be modeled as \cite{ShenITSP2012}
\begin{equation}
  \mathbf{h}_{mnk}=\hat{\mathbf{h}}_{mnk}+\mathbf{e}_{mnk},\forall m,n,k,
\end{equation}
where $\hat{\mathbf{h}}_{mnk}\in \mathbb{C}^{N_t}$ denotes the channel estimate, $\mathbf{e}_{mnk}\in \mathbb{C}^{N_t}$ is the channel estimation error which lies in a bounded set $\mathcal{E}_{mnk}=\{\mathbf{e}_{mnk}: \|\mathbf{e}_{mnk}\|_2^2\leq \delta_{mnk}^2\}$.

Under this bounded CSI error model, if the packet towards robot $k$ is transmitted in RB $m$ and OFDM symbol $n$, the worst-case received signal-to-noise ratio (SNR) of the $k$-th robot is given by
\begin{equation}\label{rho}
  \rho_{mnk}\!\!=\!\!\min_{\mathbf{e}_{mnk}\in\mathcal{E}_{mnk}}
  \frac{\left|\left(\hat{\mathbf{h}}_{mnk}\!+\!\mathbf{e}_{mnk}\right)^H\!\mathbf{w}_{mnk}\right|^2}
  {\sigma^2},\forall m,n,k.
\end{equation}
\section{Power Efficient Scheduling and Beamforming}
The typical characteristics of URLLC traffic are short-packet transmission, ultra-low latency and ultra-high reliability. The famous Polyanskiy-Poor-Verdu bound is a tight normal approximation to characterize the maximum achievable rate of short packets under AWGN channel conditions in the finite blocklength regime \cite{PolyanskiyITIT2010}. Then it has been extended to fading channels \cite{YangITIT2014}. Based on the Polyanskiy-Poor-Verdu bound and the joint channel coding scheme where one packet is encoded over all scheduled RBs, we characterize the worst-case maximum number of received data bits $R_k$ of the $k$-th robot by
\begin{align}\label{Rk}
  R_k&=\sum_{m=1}^M \sum_{n=1}^N \phi_{mnk}\log_2(1+\rho_{mnk})\notag\\
  &~~~~~~~~~~~~~-\sqrt{\sum_{m=1}^M \sum_{n=1}^N \phi_{mnk}V_{mnk}}\frac{Q^{-1}(\varepsilon_k)}{\ln 2},
\end{align}
where $V_{mnk}=1-(1+\rho_{mnk})^{-2}$ is the channel dispersion, $Q^{-1}(\varepsilon_k)$ is the inverse of Q-fucntion $Q(\varepsilon_k)=\int_{\varepsilon_k}^{\infty} \frac{1}{\sqrt{2 \pi}} e^{-t^{2}/2}dt$. In this paper, we adopt the approximation $V_{mnk}\approx 1$ based on the following two considerations: 1) the approximation is accurate enough when the received SNR is high enough, e.g. $3$ dB, as adopted in the current research works \cite{RenITC2020,SheITWC2018}; 2) this approximation actually serves as a lower bound of $R_k$, which results in a more stringent requirement.

In this paper, we are interested in the worst-case robust beamforming design and RB assignment problem under the finite blocklength regime. The objective is to minimize the total transmit power of the central controller while guarantee the specified QoS requirements $\{B_k,D_k,\varepsilon_k\}_{k=1}^K$ with imperfect CSI, which is given by
\begin{subequations}\label{P1}
\begin{align}
\mathrm{P1}:~\min_{\Phi,\mathcal{W}}~~&p_\mathrm{tot}\triangleq \sum_{k=1}^K\sum_{m=1}^M\sum_{n=1}^N \phi_{mnk}\|\mathbf{w}_{mnk}\|_2^2\label{P1a}\\
\mathrm{s.t.}~~ &R_k\geq B_k,\forall k,\label{P1b}\\
&\phi_{mnk}\in \{0,1\},\forall m,n,k,\label{P1c}\\
&\sum_{k=1}^K \phi_{mnk} \leq 1,\forall m,n,\label{P1d}\\
&\phi_{mnk}=0,\forall n>D_k,\forall m,k,\label{P1e}\\
&\|\mathbf{w}_{mnk}\|_2^2\!\leq\!\phi_{mnk}P_{\max},\forall m,n,k,\label{P1f}
\end{align}
\end{subequations}
where $\Phi=\{\phi_{mnk},\forall m,n,k\}$ is the set of binary scheduling variables and $\mathcal{W}=\{\mathbf{w}_{mnk},\forall m,n,k\}$ is the set of beamforming vectors. The constraint \eqref{P1b} means that for any robot $k$, the received number of data bits $R_k$ has to reach the target payload demand of $B_k$ information bits under all possible CSI errors. Towards this end, we have to implement the worst-cast design. Constraints \eqref{P1c} and \eqref{P1d} require that each RB is allocated to at most one robot. The delay requirement that the packet towards robot $k$ has to be successfully transmitted within $D_k$ OFDM symbols is given by \eqref{P1e}. The constraint \eqref{P1f} is the maximum transmit power constraint and guarantees that the corresponding power of beamformer is zero if $\phi_{mnk}=0$.

The problem $\mathrm{P1}$ is a mixed-integer non-convex problem. Its main challenges lie in the binary scheduling variables, the coupling structure of optimization variables, and the infinite number of strictly non-convex constraints due to the channel uncertainty in \eqref{P1b}. In general, this problem is NP-hard and seeking for the globally optimal solution by exhaustive search method will lead to extremely high computational complexity. Hence, it is necessary to develop a low-complexity algorithm to approximately solve the problem $\mathrm{P1}$ such that the robust beamforming design and RB assignment can be efficiently performed.
\section{Robust Transmission Design}
In this section, we propose a robust transmission scheme to obtain an at least sub-optimal solution for the problem $\mathrm{P1}$. First, by virtue of the binary nature of scheduling variables, we introduce a slack variable to replace the coupled term $\phi_{mnk}\mathbf{w}_{mnk}$. Furthermore, based on the independence of each robot's beamformer design, we present Lemma \ref{lemma1} to equivalently transform the original beamforming design problem into a power control problem. Eventually, $\ell_0$-norm constraints are added to guarantee the sparsity scheduling when binary scheduling variables are relaxed into continuous ones. A novel NCP approach is applied to deal with such $\ell_0$-norm constraints.

Now, let us proceed to the details. In view of the coupled nature of variables $\phi_{mnk}$ and $\mathbf{w}_{mnk}$ in the problem $\mathrm{P1}$, we introduce a slack variable $\mathbf{s}_{mnk}$ to replace $\phi_{mnk}\mathbf{w}_{mnk}$. Thus, from \eqref{Rk}, the maximum number of received data bits can be equivalently expressed as
\begin{equation}
  \bar{R}_k\triangleq\sum_{m=1}^M \sum_{n=1}^N \phi_{mnk}\log_2\left(1+\frac{\bar{\rho}_{mnk}}{\phi_{mnk}}\right)-\sqrt{l_k}\frac{Q^{-1}(\varepsilon_k)}{\ln 2},
\end{equation}
where $l_k=\sum_{m=1}^M \sum_{n=1}^N \phi_{mnk}$ and
\begin{equation}
  \bar{\rho}_{mnk}\!=\!\min_{\mathbf{e}_{mnk}\in\mathcal{E}_{mnk}}\frac{\left|\left(\hat{\mathbf{h}}_{mnk}+
  \mathbf{e}_{mnk}\right)^H\mathbf{s}_{mnk}\right|^2}{\sigma^2}.
\end{equation}
Accordingly, the total transmit power of the central controller can be alternatively given by
\begin{equation}
  \bar{p}_\mathrm{tot}\triangleq\sum_{k=1}^K\sum_{m=1}^M\sum_{n=1}^N \|\mathbf{s}_{mnk}\|_2^2.
\end{equation}

Based on the above transformations, the problem $\mathrm{P1}$ can be reformulated as
\begin{subequations}\label{P2}
\begin{align}
\mathrm{P2}:~\min_{\Phi,\mathcal{S}}~~&\bar{p}_\mathrm{tot}\label{P2a}\\
\mathrm{s.t.}~~ &\bar{R}_k\geq B_k,\forall k,\label{P2b}\\
&\|\mathbf{s}_{mnk}\|_2^2\leq\phi_{mnk}P_{\max},\forall m,n,k,\label{P2c}\\
&\eqref{P1c}-\eqref{P1e},
\end{align}
\end{subequations}
where $\mathcal{S}=\{\mathbf{s}_{mnk},\forall m,n,k\}$.

Now, the key challenge of the problem $\mathrm{P2}$ is to tackle the infinite number of strictly non-convex constraints in \eqref{P2b} and binary variable constraints in \eqref{P1c}.

First, for a given set of CSI estimation errors $\{\mathbf{e}_{mnk},\forall m,n,k\}$, the constraint \eqref{P2b} is highly intractable. Note that $\bar{R}_k$ is strictly increasing in $\bar{\rho}_{mnk}$. Thus, a standard routine \cite{JoshiITC2015} is to introduce an auxiliary variable as the lower bound of $\bar{\rho}_{mnk}$, and then apply S-procedure to equivalently convert this constraint into a linear matrix inequality (LMI) constraint. Through some optimization techniques such as SCA and semidefinite relaxation (SDR), we can eventually transform the original problem into a semidefinite program (SDP). However, it is well-known that the complexity of SDP is relatively high. Hence, in this work, by exploiting the property of the independence of each robot's beamformer design, we simplify the original beamforming design into a power control design. More specifically, it is depicted in Lemma \ref{lemma1}.
\begin{lemma}\label{lemma1}
The worst-case $\bar{\rho}_{mnk}$ of $\mathrm{P2}$ is given by
\begin{equation}
  \bar{\rho}_{mnk}=\frac{p_{mnk}\left({\|\hat{\mathbf h}_{mnk}\|_2} -\delta_{mnk}\right)^2}{\sigma^2},
\end{equation}
with
\begin{equation}
  \mathbf{s}_{mnk}^*=\sqrt{p_{mnk}}\frac{\hat{\mathbf{h}}_{mnk}}{\|\hat{\mathbf{h}}_{mnk}\|_2},
\end{equation}
where $p_{mnk}\geq0$ is the power of the beamformer $\mathbf{s}_{mnk},\forall m,n,k$.
\end{lemma}
\begin{proof}
From (7), one can readily observe that $\bar R_k$ is strictly increasing in $\bar\rho_{mnk}$, which clearly indicates that for any fixed-power beamformer, say, $\|{\mathbf s}_{mnk}\|_2^2=p_{mnk}$, the optimal beamformer to $\mathrm{P2}$ should be designed to maximize the achievable worst-case SNR, namely,
\begin{equation}\label{maxmin}
\begin{split}
  \mathbf {s}_{mnk}=\arg\max_{{\mathbf s}_{mnk}}~~&\bar \rho_{mnk}\\
  \mathrm{s.t.}\quad &\|{\mathbf s}_{mnk}\|_2^2=p_{mnk},
\end{split}
\end{equation}
for all $m,n,k$.

Based on the max-min property, we have
\begin{subequations}
\begin{align}
&\max_{{\mathbf s}_{mnk}}\min_{{\mathbf e}_{mnk}}~\left|(\hat{\mathbf h}_{mnk}+{\mathbf e}_{mnk})^H{\mathbf s}_{mnk}\right|^2\notag\\
&~~~~~~~~~~\leq \min_{{\mathbf e}_{mnk}} \max_{{\mathbf s}_{mnk}}~\left|(\hat{\mathbf h}_{mnk}+{\mathbf e}_{mnk})^H{\mathbf s}_{mnk}\right|^2 \label{maxmin1}\\
&~~~~~~~~~~=\min_{{\mathbf e}_{mnk}}~ p_{mnk}\|\hat{\mathbf h}_{mnk}+{\mathbf e}_{mnk}\|_2^2 \label{maxmin2}\\
&~~~~~~~~~~=p_{mnk}\left({\|\hat{\mathbf h}_{mnk}\|_2} -\delta_{mnk}\right)^2,
\end{align}
\end{subequations}
where \eqref{maxmin1} holds with ${\mathbf s}_{mnk}={\mathbf s}_{mnk}^*\triangleq\sqrt{p_{mnk}}(\hat{\mathbf h}_{mnk}+{\mathbf e}_{mnk})/\|\hat{\mathbf h}_{mnk}+{\mathbf e}_{mnk}\|_2$,  while \eqref{maxmin2} holds with ${\mathbf e}_{mnk}={\mathbf e}_{mnk}^*\triangleq -\delta_{mnk}\hat{\mathbf h}_{mnk}/\|{\mathbf h}_{mnk}\|_2$. And we can easily show that the equality in \eqref{maxmin} holds true with ${\mathbf s}_{mnk}={\mathbf s}_{mnk}^*$ and ${\mathbf e}_{mnk}={\mathbf e}_{mnk}^*$. This completes the proof.
\end{proof}

Therefore, let us consider the following  equivalent problem of $\mathrm{P2}$:
\begin{subequations}
\begin{align}
\mathrm{P3}:~~\min_{\Phi,\mathcal{P}}~&\tilde{p}_\mathrm{tot}\triangleq \sum_{k=1}^K\sum_{m=1}^M\sum_{n=1}^N p_{mnk}\label{P3a}\\
\mathrm{s.t.}~~ &\sum_{m=1}^{M} \sum_{n=1}^{N} \phi_{mnk} \log _{2}\left(1+\frac{g_{mnk}p_{mnk}}{\phi_{mnk}}\right)\notag\\
&~~~~~~~~~~-\sqrt{l_k} \frac{Q^{-1}\left(\varepsilon_{k}\right)}{\ln 2}\geq B_k,\forall k,\label{P3b}\\
&0\leq p_{mnk}\leq\phi_{mnk}P_{\max},\forall m,n,k,\label{P3c}\\
&\eqref{P1c}-\eqref{P1e},
\end{align}
\end{subequations}
where $\mathcal{P}=\{p_{mnk},\forall m,n,k\}$ is the set of powers and $g_{mnk}\triangleq \left({\|\hat{\mathbf h}_{mnk}\|_2} -\delta_{mnk}\right)^2/\sigma^2$.

Note that the constraint \eqref{P3b} is in the form of difference-of-concave function. By the first-order Taylor approximation of the concave function $\sqrt{l_k}$, we can obtain a locally tight upper bound $\frac{l_k+l_k^{(i)}}{2\sqrt{l_k^{(i)}}}$, where $l_k^{(i)}$ is the value of $l_k$ in the $i$th iteration. Based on this approximation, the constraint \eqref{P3b} can be approximated by a convex one as
\begin{align}\label{convexRk}
  &\sum_{m=1}^{M}\sum_{n=1}^{N}\phi_{mnk} \log _{2}\left(1+\frac{g_{mnk}p_{mnk}}{\phi_{mnk}}\right)\notag\\
  &~~~~~~~~~~~~~~~~~~~~~~~~-\frac{l_k+l_k^{(i)}}
  {2\sqrt{l_k^{(i)}}}\frac{Q^{-1}(\varepsilon_k)}{\ln 2}\geq B_k,\forall k.
\end{align}

Next, let us deal with the binary variables $\phi_{mnk},\forall m,n,k$. We first relax them into continuous ones, i.e., $\phi_{mnk}\in[0,1],\forall m,n,k$. Considering the requirement that each RB is allocated to at most one robot, we have the constraints
\begin{equation}\label{ell_0}
  \|\pmb{\phi}_{mn}\|_0\le 1,\forall m,n,
\end{equation}
to guarantee this sparsity requirement on each RB, where $\pmb{\phi}_{mn}\in \mathbb{R}_{+}^K$ is defined as $\pmb{\phi}_{mn}=[\phi_{mn1},\cdots,\phi_{mnK}]^T$. For the $\ell_0$-norm constraint, reweighted $\ell_1$ is a well-known method by converting this $\ell_0$-norm constraint into a weighted $\ell_1$-norm constraint in each iteration \cite{CandesJoFAA2007}. Here, we adopt a novel NCP approach to get the sparse scheduling solutions \cite{Wang2019}. To be more specific, the principle of the NCP approach is as follows. For any vector $\mathbf{x}\in\mathbb{R}^{n}$, it has at most one non-zero entry if and only if
\begin{equation}
  \|\mathbf{x}\|_a=\|\mathbf{x}\|_b,~1\leq a<b,
\end{equation}
where $\|\cdot\|_a$ and $\|\cdot\|_b$ represent $\ell_a$-norm and $\ell_b$-norm, respectively. Moreover, we have
\begin{equation}
    \|\mathbf{x}\|_a^v=\|\mathbf{x}\|_b^v,~1\leq a<b,
\end{equation}
by adding a power exponent $v>0$. So the constraint \eqref{ell_0} can be expressed in the following equivalent form
\begin{equation}
  \|\pmb{\phi}_{mn}\|_a^v=\|\pmb{\phi}_{mn}\|_b^v,~1\leq a<b.
\end{equation}
In general, $\|\pmb{\phi}_{mn}\|_a^v-\|\pmb{\phi}_{mn}\|_b^v\geq 0$ always holds for some $a,b$ with $1\leq a<b$. In this paper, we consider a smooth penalty by choosing $a=1,b=2,v=2$.

To promote the sparsity that $\pmb{\phi}_{mn}$ has at most one non-zero entry, we add a penalty term to the objective function, which is given by
\begin{equation}
  \mathcal{F}(\pmb{\phi}_{mn})=\frac{\lambda}{2}\sum_{m=1}^M \sum_{n=1}^N\left(\|\pmb{\phi}_{mn}\|_1^2-\|\pmb{\phi}_{mn}\|_2^2\right),
\end{equation}
where $\lambda>0$ is a penalty factor. Note that the penalty term is in the form of difference-of-convex function. To render the objective function convex, we apply the first-order Taylor approximation to the convex function $\|\pmb{\phi}_{mn}\|_2^2$, which is given by
\begin{equation}
  \|\pmb{\phi}_{mn}\|_2^2\approx 2 \left(\pmb{\phi}_{mn}^{(i)}\right)^T\!\pmb{\phi}_{mn}-\|\pmb{\phi}_{mn}^{(i)}\|_2^2,
\end{equation}
where $\pmb{\phi}_{mn}^{(i)}$ is the value of $\pmb{\phi}_{mn}$ in the $i$th iteration. Hence, the penalty term can be correspondingly approximated by
\begin{equation}
  \mathcal{F}(\pmb{\phi}_{mn})\!\approx\!\frac{\lambda}{2}\!\sum_{m=1}^M \sum_{n=1}^N\left(\!\|\pmb{\phi}_{mn}\|_1^2\!-\!
2\left(\!\pmb{\phi}_{mn}^{(i)}\right)^T\!\pmb{\phi}_{mn}\!+\!\|\pmb{\phi}_{mn}^{(i)}\|_2^2\!\right).
\end{equation}

Based on this NCP approach, we have the following penalized convex problem
\begin{subequations}
\begin{align}
\mathrm{P4}:~~\min_{\Phi,\mathcal{P}}~&\tilde{p}_\mathrm{tot}+\mathcal{F}(\pmb{\phi}_{mn})\label{P4a}\\
\mathrm{s.t.}~ &0\leq \phi_{mnk}\leq 1,\forall m,n,k,\label{P4b}\\
&\eqref{P1d}-\eqref{P1e},\eqref{P3c},\eqref{convexRk}.
\end{align}
\end{subequations}

According to the preceding analysis, now we propose a penalized SCA based iterative algorithm to solve the problem $\mathrm{P1}$. The algorithm is summarized in Algorithm \ref{Alg1}. Note that a sequence of the penalized problem $\mathrm{P4}$ can be efficiently solved, which can yield a stationary-point solution of the problem $\mathrm{P1}$ after convergence \cite{Wang2019}.
\begin{algorithm}[t]
\small
\caption{\small \textbf{:} Penalized SCA based Algorithm for Solving $\mathrm{P1}$}
\label{Alg1}
\begin{spacing}{1.2}
	\begin{algorithmic}[1]
		\STATE \textbf{Initialize}~iteration index $i=0$, feasible $\phi_{mnk}^{(0)},\forall m,n,k$, initial penalty factor $\lambda^{(0)}>0$, $\eta>1$, tolerance $\epsilon>0$.
		\STATE Calculate $l_k^{(0)} =\sum_{m=1}^{M} \sum_{n=1}^{N} \phi_{m n k}^{(0)},\forall k$.
		\REPEAT
        \STATE Set $i=i+1$.
		\STATE Obtain $\{\phi_{mnk},p_{mnk}\}$ by solving $\mathrm{P4}$ and restore $\tilde{p}_\mathrm{tot}^{(i)}$.
        \STATE Update $\phi_{mnk}^{(i)}=\phi_{mnk},\forall m,n,k$.			
		\STATE Update $l_k^{(i)}=\sum_{m=1}^M\sum_{n=1}^N \phi_{mnk}^{(i)},\forall k$.
        \STATE Update $\lambda^{(i)}=\eta \lambda^{(i-1)}$.
		\UNTIL $|\tilde{p}_\mathrm{tot}^{(i)}-\tilde{p}_\mathrm{tot}^{(i-1)}|\leq \epsilon$ and $\mathcal{F}(\pmb{\phi}_{mn})\leq \epsilon$.
		\STATE \textbf{Output} $\phi_{mnk}$ and $p_{mnk}$ for all $m,n,k$.
	\end{algorithmic}
\end{spacing}
\end{algorithm}
Based on the solutions $\Phi,\mathcal{P}$, we can easily recover the beamformer $\mathbf{w}_{mnk}$ of the original problem $\mathrm{P1}$ by $\sqrt{p_{mnk}}\hat{\mathbf h}_{mnk}/\|\hat{\mathbf h}_{mnk}\|_2$ for all $m,n,k$.
\section{Numerical Results}
In this section, we present some numerical results to evaluate the performance of the proposed robust transmission scheme.

In the simulation, $d_k$ in meter denotes the distance from the controller to the robot $k$. We set $d_1=100,d_2=240,d_3=180,d_4=300$. The small-scale fading components of all channel estimates are assumed to be independent and identically distributed circularly symmetric complex Gaussian random variables with zero mean and unit variance. We generate $100$ channel realizations and take their average as the simulation result. Other parameters are listed in Table \ref{table1}, unless otherwise specified.
\begin{table}[htbp]\footnotesize{
\centering
\caption{Simulation Parameters}
\label{table1}
\setlength{\tabcolsep}{1.2mm}{
\begin{tabular}{lll}
  \toprule
  Symbol & Parameter & Value\\
  \midrule
  $K$ & Number of robots & $4$\\
  $M$ & Number of RBs & $10$\\
  $N$ & Number of OFDM symbols & $6$\\
  $N_t$ & Number of transmit antennas & $2$\\
  $P_{\max}$ & Maximum transmit power & $30$ dBm\\
  $W$ &Each RB bandwidth &180 kHz\\
  $\varepsilon_k\!\!=\!\!\varepsilon,\forall k$ & Packet error probability & $10^{-6}$\\
  $B_k\!=\!B,\forall k$ & Number of data bits & $40$\\
  $\delta^2$ & CSI error bound & $0.01$\\
  $D_k,\forall k$ & Delay of robots & $D_1\!\!=\!\!D_2\!\!=\!\!2,D_3\!\!=\!\!3,D_4\!\!=\!\!4$\\
  $N_0$ & Noise power spectral density & $-173$ dBm/Hz\\
  $\mathrm{PL}_k$ & Path loss model (dB) & $35.3+37.6\log_{10}(d_k)$\\
  $\lambda^{(0)}$ & Initial penalty factor & $0.001$\\
  $\eta$ & Scaling factor & $1.8$\\
  $\xi$ & Parameter in reweighted $\ell_1$ & $0.01$\\
  \bottomrule
\end{tabular}}}
\end{table}

Fig. \ref{PerformanceContrast} compares the performance of NCP and reweighted $\ell_1$ in terms of power consumption and convergence rate. As can be seen from Fig. \ref{PerformanceContrast}(a), under $100$ channel realizations, the performance of NCP-based robust resource allocation is always the same or superior to the reweighted $\ell_1$ method. Moreover, in Fig. \ref{PerformanceContrast}(b), the convergence rate of NCP is significantly faster than that of reweighted $\ell_1$, which is an extremely important factor for the delay-sensitive applications especially for URLLC scenarios. In addition, the convergence performance of the NCP-based algorithm is more robust to channel realizations than that of the reweighted $\ell_1$ method.
\begin{figure}[!t]\centering
	\subfloat[$\tilde{p}_\mathrm{tot}$ under $100$ channel realizations]{\includegraphics[width=.63\linewidth]{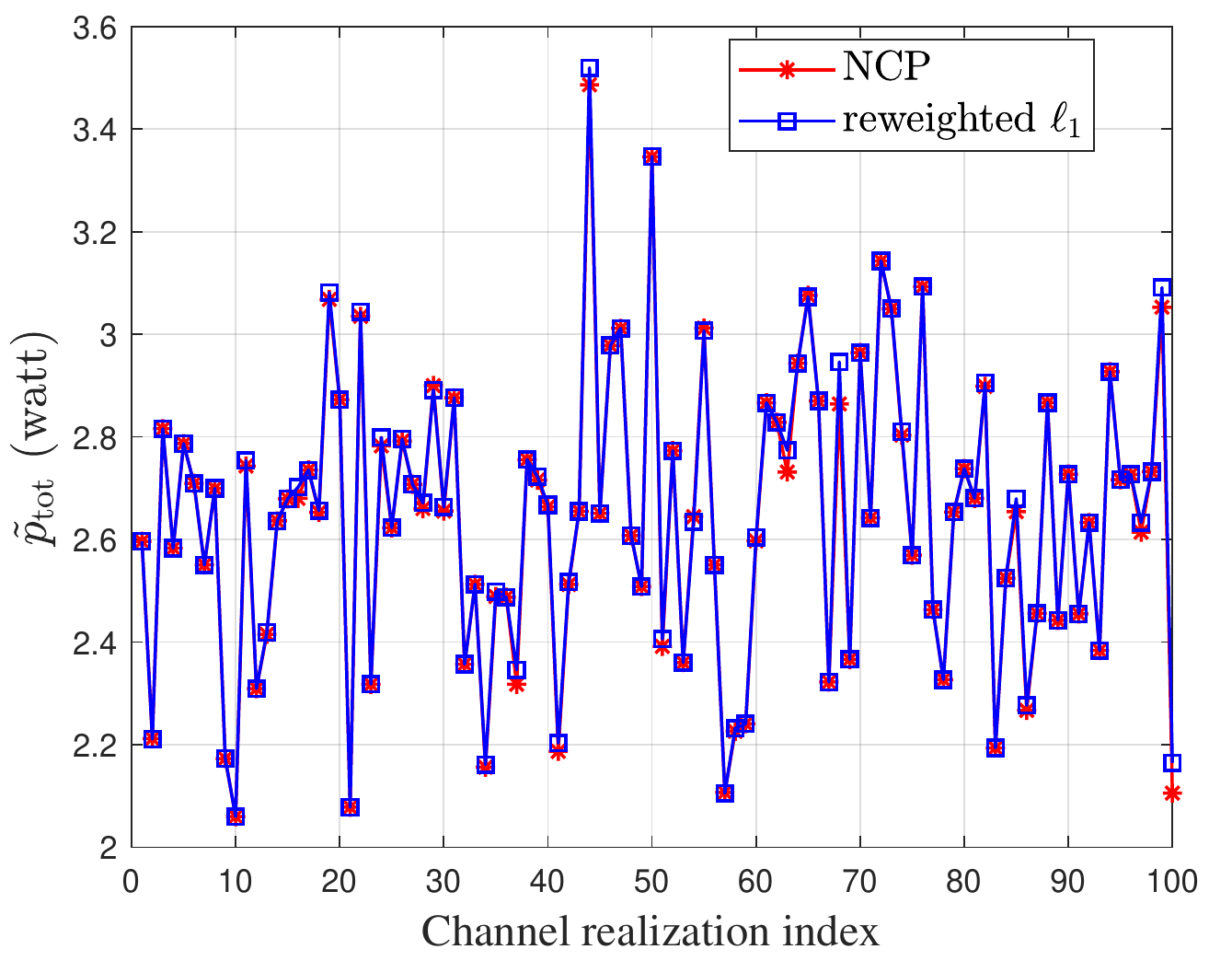}}\\[0.01mm]
	\subfloat[Number of iterations under $100$ channel realizations]{\includegraphics[width=.63\linewidth]{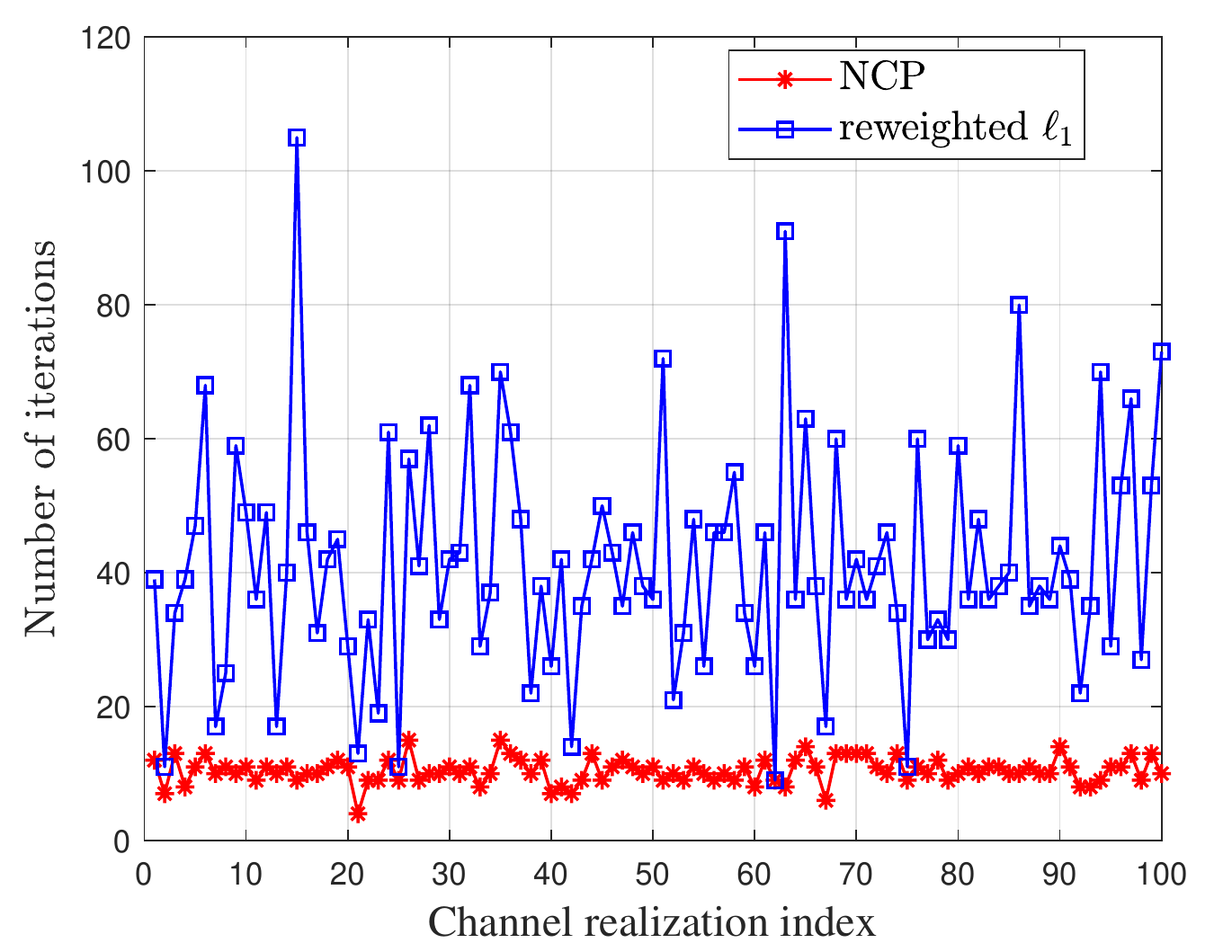}}
	\caption{Performance comparison between NCP and reweighted $\ell_1$.}
	\label{PerformanceContrast}
\end{figure}

Fig. \ref{initialPenalty} presents the convergence performance of the NCP approach under different values of the initial penalty factor $\lambda^{(0)}$ and scaling factor $\eta$. Note that the power initially decreases then increases and finally remains constant. The reason is that in the first few iterations, the penalty factor is small ($\ell_0$-norm constraints do not work), so the total transmit power $\tilde{p}_\mathrm{tot}$ will be reduced; as the number of iterations increases, the penalty factor becomes larger, and $\ell_0$-norm constraints begin to work, forcing the variable $\Phi$ to be sparse, so $\tilde{p}_\mathrm{tot}$ starts to increase, and eventually becomes stable. In addition, the initial values of $\lambda^{(0)}$ and $\eta$ result in different convergence performance, but they can all converge quickly.
\begin{figure}[!t]\centering
	\includegraphics[width=.63\linewidth]{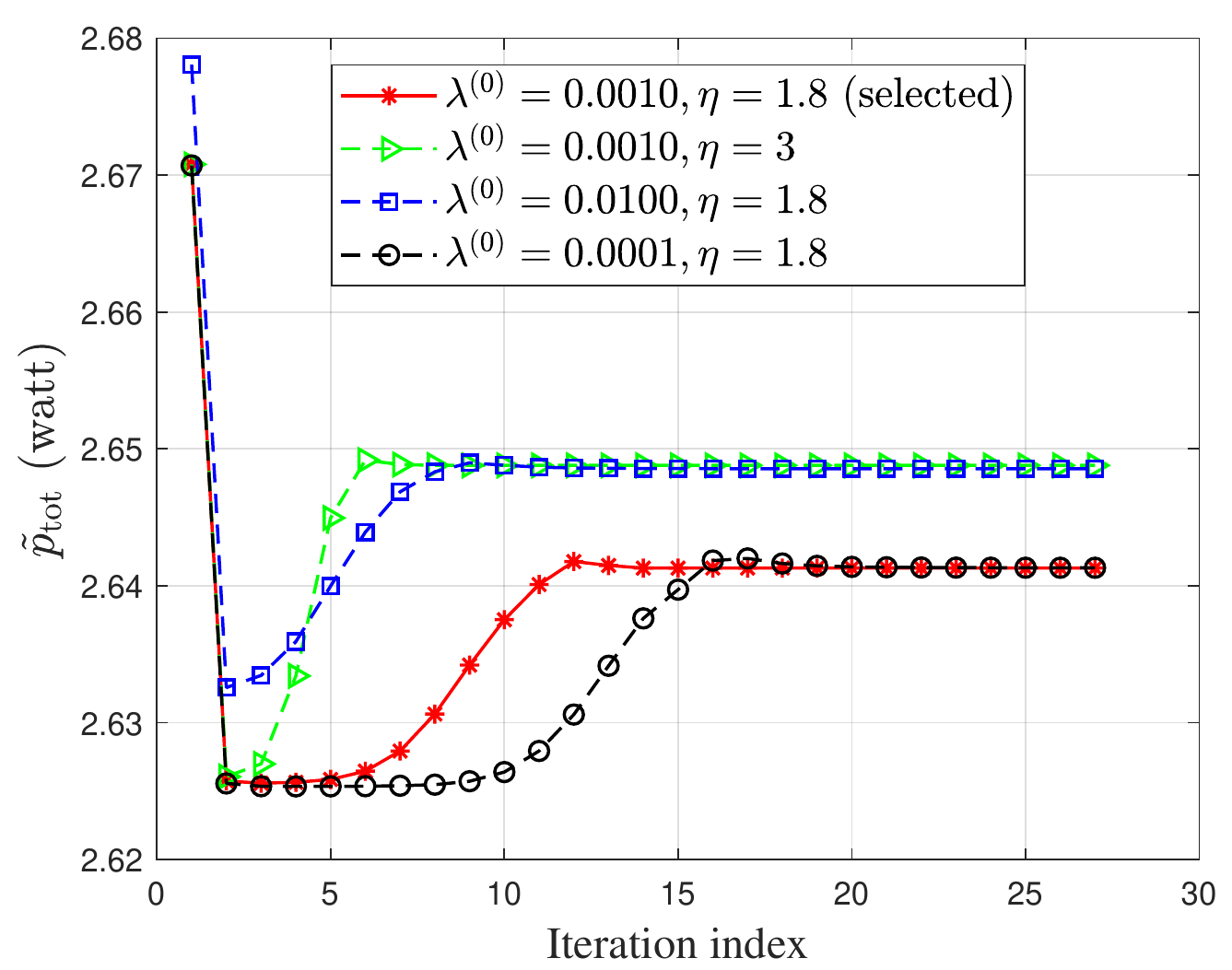}
	\caption{Convergence performance of the NCP-based algorithm under different values of $\lambda^{(0)}$ and $\eta$.}
	\label{initialPenalty}
\end{figure}

In Fig. \ref{Ptot_vs_reliability}, we plot the required total transmit power $\tilde{p}_\mathrm{tot}$ versus packet error probability $\varepsilon$ for different values of $N_t,\delta^2$ and $D_1$. As expected, $\tilde{p}_\mathrm{tot}$ monotonically decreases with the increase of allowable packet error probability $\varepsilon$. Also, the increase of the number of transmit antennas, the improvement of channel estimation accuracy and the relaxation of the delay requirement can reduce the total power consumption. Among them, the number of transmit antennas dominates the impact on power consumption. In Fig. \ref{Ptot_vs_reliability}(a), we note that as $N_t$ increases, the impact of the channel estimation error $\delta^2$ on $\tilde{p}_\mathrm{tot}$ becomes significantly small. This is due to the fact that more diversity gain provided by multiple antennas can compensate the effect of the imperfect channel estimation. This demonstrates that multiple antennas play a critical role for communication reliability, thereby reducing power consumption. In Fig. \ref{Ptot_vs_reliability}(b), there is an interesting finding that the solid line labeled with $D_1=2,\delta^2=0.01$ nearly coincides with the dashed line labeled with $D_1=4,\delta^2=0.05$. This indicates that the impact of shortening the delay to half can be offset by increasing the accuracy of channel estimation to $5$ times the original, which provides some interesting insights for the future URLLC research.
\begin{figure}[!t]\centering
	\subfloat[Effect of $N_t$ and $\delta^2$  with fixed $D_1=2$]{\includegraphics[width=.64\linewidth]{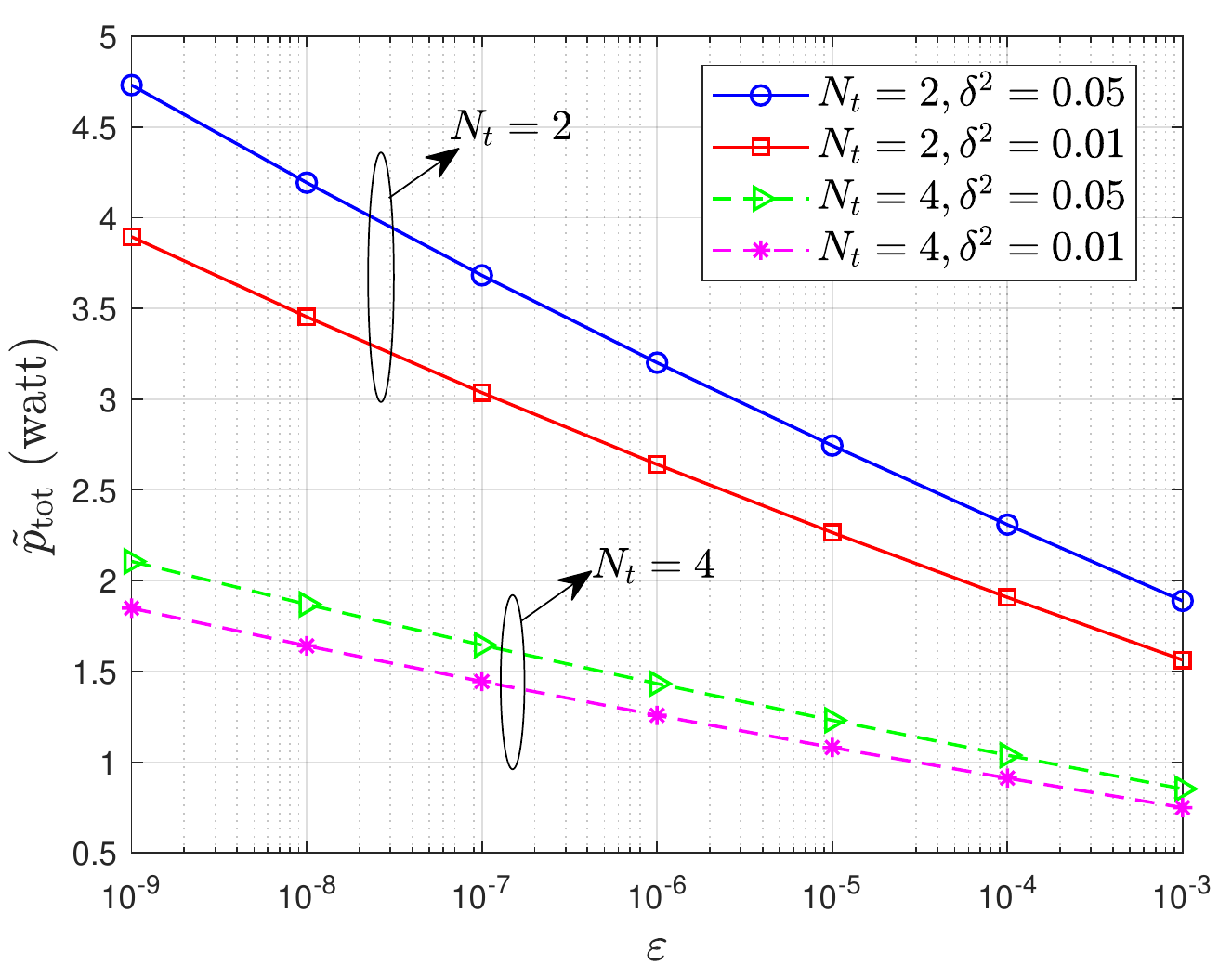}}\\[0.01mm]
	\subfloat[Effect of $D_1$ and $\delta^2$ with fixed $N_t=2$]{\includegraphics[width=.64\linewidth]{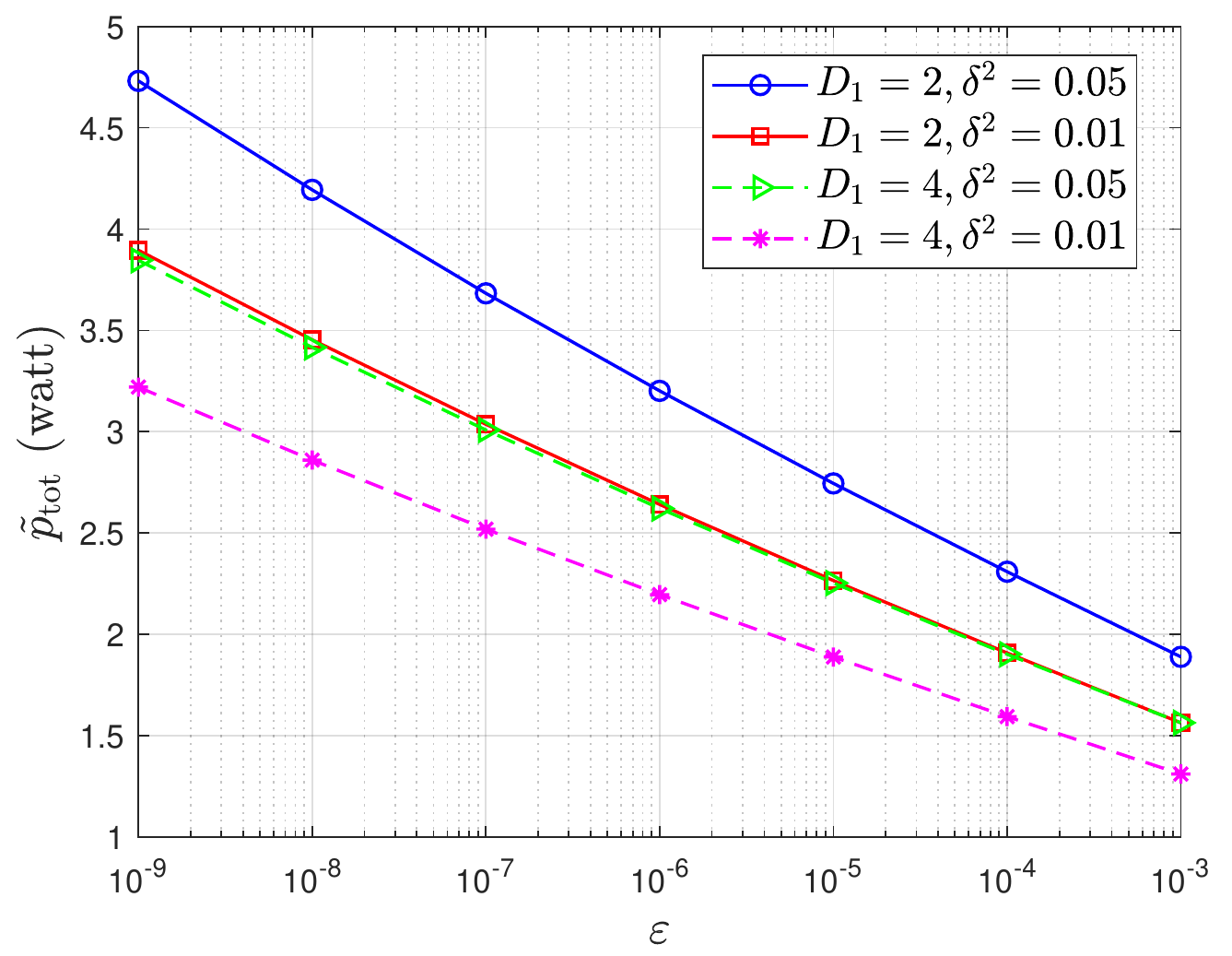}}
	\caption{$\tilde{p}_\mathrm{tot}$ versus $\varepsilon$ for different values of $N_t$, $\delta^2$ and $D_1$.}
	\label{Ptot_vs_reliability}
\end{figure}
\section{Conclusion}
In this work, we propose a robust power-efficient RB assignment scheme to guarantee each robot's URLLC requirements in a MISO-OFDMA system. The formulated mixed-integer robust design problem is resolved in two steps. For the worst-case SNR caused by the channel uncertainty, we leverage the property of the independence of each robot's beamformer design, and thereby obtain an equivalent joint design problem of power control and RB assignment. For the binary RB indicator constraints, we propose a novel NCP approach to guarantee the sparsity on each RB. The proposed penalized SCA based algorithm can yield a stationary-point solution. Numerical results show that the performance of NCP-based resource scheduling is always the same or superior to the well-known reweighted $\ell_1$ method. Also, NCP performs much faster than reweighted $\ell_1$ and is more robust to channel realizations. We also investigate the impacts of latency, reliability, number of transmit antennas and channel uncertainty on the system performance, which shed some light on the zero-delay URLLC under 6G. The cases of robust MIMO-URLLC transmission and efficient prediction of channels and traffic for near-zero latency by machine learning will be studied in the future work.

\end{document}